\documentclass[letterpaper, 10 pt, conference, final, twocolumn]{ieeeconf}  

\IEEEoverridecommandlockouts                              

\usepackage{cite}
\usepackage{amsmath,amssymb,amsfonts}
\usepackage{graphicx}
\usepackage{textcomp}
\usepackage{color}
\usepackage{bbm}

\usepackage{balance}

\newtheorem{thm}{Theorem}

\newtheorem{assumption}{Assumption}

\newtheorem{exx}{Example}
\newtheorem{remm}{Remark}

\newenvironment{assume}{\begin{assumption}\rm }{\hfill \hspace*{1pt} \hfill $\lrcorner$ \end{assumption}}
\newenvironment{theorem}{\begin{thm}\rm }{\hfill \hspace*{1pt} \hfill $\lrcorner$ \end{thm}}
\newenvironment{remark}{\begin{remm}\rm }{\hfill \hspace*{1pt} \hfill $\lrcorner$\end{remm}}

\newcommand\real{\ensuremath{{\mathbb R}}}
\newcommand\realn{\ensuremath{{\mathbb{R}^n}}}

\newcommand\tran{\ensuremath{\mathsf{T}}}

\newcommand\mymatrix[2]{\left[\begin{array}{#1} #2 \end{array}\right]}

\newcommand{\calD}{\mathcal{D}}
\newcommand{\calE}{\mathcal{E}}

\newcommand{\calH}{\mathcal{H}}

\newcommand{\calL}{\mathcal{L}}

\newcommand{\calP}{\mathcal{P}}

\newcommand{\calS}{\mathcal{S}}

\newcommand{\grad}{\ensuremath\operatorname{grad}}
\newcommand{\bgrad}{\ensuremath\mathbf{{grad}}}

\newcommand{\dd}[1]{\ensuremath \mathrm{d}{#1}}

\newcommand\inner[2]{\left\langle #1 , #2 \right\rangle}

\newcommand{\bv}{\ensuremath\mathbf{v}}
\newcommand{\bi}{\ensuremath\mathbf{i}}
\newcommand{\bV}{\ensuremath\mathbf{V}}
\newcommand{\bI}{\ensuremath\mathbf{I}}
\newcommand{\bg}{\ensuremath\mathbf{g}}
\newcommand{\bG}{\ensuremath\mathbf{G}}
\newcommand{\bD}{\ensuremath\mathbf{D}}

\usepackage{bbold}
\newcommand{\one}{\ensuremath \mathbb{1}}

\title{\LARGE \bf
	Gradient modelling of memristive systems
}

\author{Fulvio Forni$^{1}$ and Rodolphe Sepulchre$^{2,1}$
	\thanks{*The research leading to these results has received funding from the European Research Council under the
		Advanced ERC Grant Agreement SpikyControl n.101054323.}
	\thanks{$^{1}$Department of Engineering, University of Cambridge, Trumpington Street, Cambridge CB2 1PZ, United Kingdom
		 {\tt\small f.forni@eng.cam.ac.uk}}%
	\thanks{$^{2}$Department of Electrical Engineering, KU Leuven, KasteelPark Arenberg, 10, B-3001 Leuven, Belgium {\tt\small rodolphe.sepulchre@kuleuven.be}} 
	}

\begin{document}

\maketitle
\thispagestyle{empty}
\pagestyle{empty}

\begin{abstract}
	We introduce a gradient modeling framework for memristive systems. Our focus is on memristive systems as they appear in neurophysiology and neuromorphic systems. Revisiting the original definition of Chua, we regard memristive elements as gradient operators of quadratic functionals with respect to a metric determined by the memristance. We explore the consequences of gradient properties for the analysis and design of neuromorphic circuits.
\end{abstract}


\section{introduction}

Memristive elements -- resistive elements with `memory' -- lie at the core of neuronal modelling and neuromorphic engineering. The ideal element $i=g(q)v$, a conductor whose conductance (or rather, \textit{mem}ductance) depends on charge, was first introduced by Chua as a theoretical construct motivated by a `missing' element next to the ideal resistor, capacitor, and inductor \cite{chua1971}. A few years later, Chua and Wang \cite{chua1976} made the memristor a special case of a broad family of circuit elements called memristive systems. Next to the ideal memristor, memristive systems include the conductance-based models of ion channels introduced in the seminal work of Hodgkin and Huxley \cite{hodgkin1952} to model neuronal excitability. More generally, any biophysical model of neuronal circuits consists of the parallel interconnection of capacitors with a bank of memristive elements and constant voltage sources (batteries). In that sense, memristive elements lie at the core of biophysical modeling in neuroscience.

The ideal memristor remained a theoretical construct until Hewlet-Packard announced to have built a device with memristive properties \cite{strukov2008}. Over the last decade, much research has been devoted to the design of new materials with memristive properties \cite{song2023}. Memristive elements are expected to play a central role in the future development of  neuromorphic --or more generally "in memory" -- computing \cite{yang2013}. 

Despite the importance of memristive elements in both neurophysiology and neuromorphic engineering, we currently lack a theory that includes memristive elements  in circuit design next to the traditional RLC elements.  An exception is the recent PhD dissertation \cite{huijzer2025PhD} and the publications \cite{huizer2023,huijzer2025} that consider behaviors of memristive networks from a system perspective, that is, the system properties of interconnected memristive elements.

Motivated by analysis and design questions in neuromorphic circuit design and control \cite{ribar2021,sepulchre2022}, this article revisits the original model of \cite{chua1976} to highlight the gradient nature of memristive systems. We aim at approaching the physical modeling of memristive systems as a generalization of resistive systems. The gradient properties of resistors have been at the core of extending the linear of RLC circuits to circuits with nonlinear resistors. This approach goes back to the seminal work by Brayton and Moser \cite{brayton1964}, but is also extensively developed in the systems and control literature. Gradient-based modelling of physical models was initially explored by Brockett \cite{brockett1977}, who developed the concept of input-output gradient systems (or input-output Hamiltonian systems). Key references include the comprehensive article \cite{Jeltsema2009}, as well as Chapter 14 in the monograph \cite{vanderschaft2014}.

The proposed angle of attack in this paper is to regard memristive systems not as gradient vector fields in the state-space, but instead as gradient operators in the space of past trajectories. The gradient operator is a relationship between the past current and voltage trajectories of the memristive element. It reduces to the familiar gradient model of a resistive element in the limit of a memoryless element, that is, when the past reduces to the present.  When memristive elements are connected to capacitors, we show that they lead to gradient behaviors, in the sense that (past) input-output trajectories can be regarded as critical points of a gradient flow in the space of trajectories. In the memoryless limit of resistive elements, we recover the classical gradient system property of RC circuits, that is, the state-space behavior is a gradient flow. 

The paper is organised as follows.  Section 2 revisits the original definition of Chua and Wang \cite{chua1976} to highlight two key properties of memristive systems: fading memory and dissipation. Section 3 proceeds with the reformulation of a memristive model as the gradient of a quadratic functional with respect to a Riemannian metric determined by the memristance. Section 4 extends the gradient modelling framework to memRC circuits while Section 5 illustrates the proposed modelling framework on the celebrated model of Hodgkin and Huxley.

\section{Memristive systems} \label{behaviours}
\subsection{Definition and examples}
We adopt the definition of Chua and Wang \cite{chua1976}: a voltage-controlled memristive one-port is represented by the state-space model
\begin{subequations}
\label{eq:memristive}
\begin{align}
    \dot x = f(x,v)  \\
    i = g(x,v)v
\end{align}
\end{subequations}
where $g(x,v)$ is called the {\it memconductance} of the element. The state-space equation models the voltage-dependent memory of the memconductance. 

Memristive systems generalize the memristor introduced by Chua in \cite{chua1971}, with state-space model $\dot x = v$. One motivation to introduce memristive system was to include the memristor as an ideal element of a broader family of electrical elements that would include the ionic currents introduced by Hodgkin and Huxley in \cite{hodgkin1952}. 

\begin{exx}
[Potassium current model]
The potassium current of Hodgkin and Huxley model \cite{hodgkin1952} is
\begin{align}
    \dot n &= \alpha(v)n + \beta(v)(1-n) \nonumber \\
    i &= g_K n^4 (v-v_K)
\end{align}
The parameter $g_K$ is the maximal conductance. The state $n \in [0,1]$ is called a gating variable. The state $n=0$ models the closed state of the gate whereas the state $n=1$ models the fully open state of the gate. The state equation can be regarded as the mean-field kinetic equation that models stochastic transitions between the open and closed state, with voltage dependent probability of opening ($\alpha$) and closing ($\beta)$. The parameter $v_K$ is called the Nernst potential of the potassium current. it can be regarded as a constant battery in series with the memristive element.   \hfill $\lrcorner$
\end{exx}

The memristive nature of the potassium ionic current is a general property of conductance-based neuronal models. A neuronal model includes possibly many different types of ionic currents, all obeying the structure of a memristive element in series with a constant battery.

Synaptic currents are also memristive. They model interconnections between a presynaptic neuron $v_{pre}$ and a postsynaptic neuron $v_{post}$. 
\begin{exx}[Synaptic current model]
A common model of synaptic current has the expression
\begin{align}
    \dot s &= \alpha(v_{pre})s + \beta(v_{pre})(1-s) \nonumber \\
    i & = g_{syn} s (v_{post}-v_{syn})
\end{align}
The memristive element is of the type 
\begin{align}
\left ( \begin{array}{c} i_{pre} \\ i_{post}\end{array} \right )=\left ( \begin{array}{cc} 0 & 0 \\ 0 & g_{syn}s \end{array} \right ) \left ( \begin{array}{c} v_{pre} \\ v_{post}\end{array} \right ) 
\end{align}  \hfill $\lrcorner$
\end{exx} 

The neuronal model of a general conductance-based model is the parallel interconnection of a leaky capacitor modelling the passive membrane with a bank of internal and external ionic currents, each modelled as the series interconnection of a memristive element with a battery. This is our primary source of interest for memristive modelling in this paper.

\subsection{Fading memory}

Fading memory is a key property that was formalized by Boyd and Chua in later work in \cite{Boyd1985} but guarantees the {\it well-posedness} of state-space memristive model to ensure a number of input-output properties analyzed in \cite{chua1976}: DC characteristics, limiting linear characteristics, small signal characteristics, etc.

Fading memory ensures that the memory of the memconductance is fading with time, that is, does not depend on the distant past of the voltage. Because of  time-invariance and causality of the state-space model, it is sufficient to define the continuity property for the memory functional from the past voltage to the current output. Fading memory assumes continuity of this functional with respect to faded past signals. 

\begin{assume}\emph{[Fading memory]}
\label{assumption:fading}
The state-space model (\ref{eq:memristive}) defines a memory functional 
$g_0 \in  \calL_{2(-\infty,0]}  \rightarrow \real$  on the space of past signals $\bv \in \calL_{2(-\infty,0]}$ 
\begin{align}
\bv \rightarrow g_0(\mathbf{v})=g(x(0),v(0)).
\end{align}

The memory functional has (exponential) fading memory, that is, it is continuous in the space $\calL_{2(-\infty,0]}$ equipped with the weighted quadratic norm induced by the inner product
\begin{equation}
\inner{\bv_1}{\bv_2}_\lambda = 
\int_{-\infty}^0 \bv_1(t) \bv_2(t) e^{2 \lambda t} \dd{t}
\end{equation}
for some given $\lambda>0$ and generic $\bv_1,\bv_2 \in \calL_{2(-\infty,0]}$. 
\end{assume}

Consider now the signal space 
$\calL_{2(-\infty,T]}$ for some $T\in \real$.
By time-invariance, 
the memory functional $g_0$
\emph{defines} the memconductance 
$\bg \in \calL_{2(-\infty,T]} \to \calL_{\infty}$ as follows. 
For all $\bv \in \calL_{2(-\infty,T]}$ and all $t \in (-\infty,T]$,
\begin{align}
\label{eq:memfunctional2memconductance}
\bg(\bv)(t) = g_0(\Delta_{t}(\bv))
\end{align}
where $\Delta_{t}$ is the time-shift and projection 
operator 
\begin{align}
\Delta_{t}(\bv)(\cdot) = \bv(\cdot+t) \in \calL_{2(-\infty,0]}.
\end{align}

A memristive system can thus be regarded as  a voltage-dependent instantaneous mapping relating input voltage $\bv\in\calL_{2(-\infty,T]}$ to output current $\bi\in \calL_{2(-\infty,T]}$ given by
\begin{equation}
\label{eq:memristor}
\bi(t) \, = \, \bg(\bv)(t) \, \bv(t),
\end{equation}
where $t \in (-\infty,T] \subseteq \real$ and 
$\bg \in \calL_{2(-\infty,T]} \to \calL_\infty$
is the memconductance fading memory operator defined by the memory functional $g_0$.

In what follows, for simplicity of notation, we assume fading memory with $\lambda =0$, that is, with respect to the usual norm, but all the results in the paper readily extend to a weighted norm. Also, whenever possible,
we will drop the dependence on time from the equations. To avoid confusion, we will
use bold symbols for signals and operators, and nonbold symbols to denote \emph{their evaluation at a generic time} $t$. 
For instance, this allows us to write \eqref{eq:memristor} as
\begin{equation}
\label{eq:memristor2}
i \, = \, g(\bv) \, v.
\end{equation}

\eqref{eq:memfunctional2memconductance} guarantees causality of the menconductance, that is, of the voltage-current relationship. The memconductance is thus an instantaneous mapping between voltage and current that depends on the voltage \emph{past}, modeling memory effects in the circuit. 
\subsection{Dissipativity}

The memconductance of a memristive element is {\it positive}.
\begin{assume}\emph{[Energy dissipation]}
\label{assumption:causality_and_invertibility}
For all $\bv \in \calL_{2(-\infty,T]}$ 
and $t\in (-\infty,T]$,
\begin{equation}
 \label{eq:assume_invertibility} 
\bg(\bv)(t)  \geq \varepsilon > 0,
\end{equation}
for some (uniform) constant $\varepsilon$.
\end{assume}

This second assumption is key to the dissipativity properties identified in \cite{chua1976}: passivity, no energy discharge, double-valued Lissajou property, local passivity, etc.

The instantaneous power  $\bi(t)\bv(t) \ge 0 $ is nonnegative at any time. The element is passive \cite{Desoer1975}, dissipates energy at any time, and has no storage. The dissipated energy is
\begin{align}
\int_{-\infty}^T \bv^\tran(t) \bg(\bv)(t) \bv(t) \dd{t} \ge 0
\end{align}
which can be regarded as the inner product $\inner{\bv}{\bi}$
with respect to the  the inner product 
\begin{equation}
\inner{\bv_1}{\bv_2} = 
\int_{-\infty}^T \bv_1(t) \bv_2(t)  \dd{t}
\end{equation}
where $\bv_1, \bv_2$ are generic signals in $\calL_{2(-\infty,T]}$. 

From a geometric perspective, if we endow $\calL_{2(-\infty,T]}$ with the structure of a Riemannian manifold
$(\calL_{2(-\infty,T]},\bg)$, the dissipated energy of a
memristive element corresponds to
the quadratic norm $\parallel \!\!  \bv  \!\! \parallel^2_g $ associated to the weighted inner product 
\begin{equation}
\inner{\bv_1}{\bv_2}_g = \int_{-\infty}^T \bv_1(t) \bg(\bv)(t) \bv_2(t)   \dd{t} \geq 0,
\end{equation}
where $\bv_1$ and $\bv_2 \in T_\bv \calL_{2(-\infty,T]}$
are generic tangent vectors. 
This interpretation of the memconductance as a Riemannian metric is the basis of the gradient modelling in the next section.

\section{Gradient models of memristive elements}
\label{section:memristive_elements}

\subsection{Gradient models of resistive elements}
\label{section:resistive_elements_characterization}
As a first step, we revisit the gradient modelling of a resistive element, that is, when the memconductance $g$ is a constant (linear resistor) or a static function $g(v)$ (nonlinear resistor).

For a linear resistor 
\begin{align}
\label{eq:linear_resistor}
i=gv,
\end{align}
the conductance $g \in \real_{>0}$ is constant and the energy characterization is through the instantaneous (dissipated) power
\begin{subequations}
\label{eq:linear_resistor_dissipation_and_cocontent}
\begin{align}
\label{eq:linear_resistor_dissipation}
\calD(v) := \inner{i}{v} = g v^2  = \inner{v}{v}_g, 
\end{align}
for all $v \in \real$, where we have made use of the weighted inner product
$\inner{v_1}{v_2}_g =  g v_1 v_2$, for
$v_1,v_2 \in \real$.

An equivalent energy characterization is through  the resistive co-content \cite{Jeltsema2009} 
\begin{align}
\label{eq:linear_resistor_cocontent}
\calE(\bar{v}) 
:= \!\int_0^{\bar{v}}\inner{i}{\dd{v}} 
= \!\int_0^{\bar{v}}   gv \dd{v} 
= \frac{1}{2}\inner{\bar{v}}{\bar{v}}_g
\end{align}
\end{subequations}
for all $\bar{v} \in \real$.

Dissipation and co-content completely determine 
the resistive relationship between current and voltage:
given $\calD$ in \eqref{eq:linear_resistor_dissipation},
\eqref{eq:linear_resistor} is the unique solution to the 
identity $\inner{i}{v} = \calD(v)$. 
Likewise, given \eqref{eq:linear_resistor_cocontent},
the gradient of the co-content \cite{Lee1997,Absil2008} satisfies 
\begin{align}
\label{eq:linear_resistor_grad}
\inner{\grad \calE(v)}{v} = D_v \calE(v) = gv^2
\end{align}
where $D_w \calE(v)$ is the directional derivative 
of $\calE$ at $v$ in the direction of $w$. 
This leads to 
\begin{align}
\label{eq:linear_resistor_grad}
i = \grad \calE(v) = gv.
\end{align}

To extend the energy characterization to the nonlinear case
\begin{align}
\label{eq:nonlinear_resistor}
i = g(v)v,
\end{align}
we regard $g:\real \to \real$ as a Riemaniann tensor, 
leading to the Riemannian manifold $(\real,g)$ with
inner product
$\inner{v_1}{v_2}_{g} = g(v)v_1 v_2$,
where $v_1,v_2 \in T_v \real$ are generic tangent vectors at $v\in \real$.

The instantaneous power at $v\in\real$ still satisfies
\begin{subequations}
\label{eq:nonlinear_resistor_dissipation_and_cocontent}
\begin{align}
\label{eq:nonlinear_resistor_dissipation}
\calD(v) := \inner{i}{v} = g(v)v^2 = \inner{v}{v}_g, 
\end{align}
but the resistive co-content at $\bar{v} \in \real$
\begin{align}
\label{eq:nonlinear_resistor_cocontent}
\calE(\bar{v}) 
:= \!\int_0^{\bar{v}}\inner{i}{\dd{v}} 
= \!\int_0^{\bar{v}} g(v)v \dd{v} 
\neq \frac{1}{2}\inner{\bar{v}}{\bar{v}}_g 
\end{align}
\end{subequations}
is no longer equivalent to the instantaneous dissipation.

However, both continue to determine the element's behavior. 
Given $\calD$ at $v$ in \eqref{eq:nonlinear_resistor_dissipation}, 
$i=g(v)v$
is the unique solution to 
$\inner{i}{v} = \calD(v)$. Likewise,
given \eqref{eq:nonlinear_resistor_cocontent} at $v$,
\begin{align}
\label{eq:nonlinear_resistor_grad}
\inner{\grad \calE(v)}{v}= D_v\calE(v) = g(v)v^2
\end{align}
where the last identity follows from the 
fact that directional derivative of the co-content at
$\bar{v}$ in the direction $w$ is
\begin{align}
D_w \calE(\bar{v}) = 
\lim_{\epsilon \to 0}
\frac{1}{\epsilon}\int_{\bar{v}}^{\bar{v}+\epsilon w}g(v)v dv
=
g(\bar{v})\bar{v} w.
\end{align}
This leads to 
\begin{equation}
i = \grad \calE(v) = g(v)v.
\end{equation}

The linear equivalence between dissipation and co-content can however be recovered by regarding  the 
voltage space as a Riemannian manifold with metric $\frac{1}{g}$. 
As suggested in Figure \ref{fig:linearization}, the `nonlinear' area element linearizes the infitesimal relationship between current and voltage.  
At each $v\in \real$, this leads to the 
$\frac{1}{g}$-weighted dissipation 
\begin{subequations}
\label{eq:nonlinear_resistor_manifold_dissipation_and_cocontent}
\begin{align}
\label{eq:nonlinear_resistor_manifold_dissipation}
\calD_{\frac{1}{g}}(v)
:= \inner{i}{v}_{\frac{1}{g}} = v^2 = \inner{v}{v}
\end{align}
and $\frac{1}{g}$-weighted co-content
\begin{align}
\label{eq:nonlinear_resistor_manifold_cocontent}
\calE_{\frac{1}{g}}(\bar{v}) 
:= \!\int_0^{\bar{v}}\!\!\!\!\inner{i}{\dd{v}}_{\frac{1}{g}} 
= \!\int_0^{\bar{v}} \!\!\!\!  v \dd{v} 
= \frac{1}{2} \bar{v}^2
= \frac{1}{2}\inner{\bar{v}}{\bar{v}}.
\end{align}
\end{subequations}

Thanks to the Riemannian perspective, the definition of 
$\calD_{\frac{1}{g}} $ and $\calE_{\frac{1}{g}}$
is the 
same for linear and nonlinear resistors. 
The constructions
\eqref{eq:nonlinear_resistor_manifold_dissipation_and_cocontent}
are similar to \eqref{eq:linear_resistor_dissipation_and_cocontent}
and
\eqref{eq:nonlinear_resistor_dissipation_and_cocontent},
but start from the Riemannian inner product
$\inner{\cdot}{\cdot}_{\frac{1}{g}}$ to retain the quadratic expression
$\inner{v}{v}$.

\begin{figure}[htbp]
    \centering
    \includegraphics[width=.94\columnwidth]{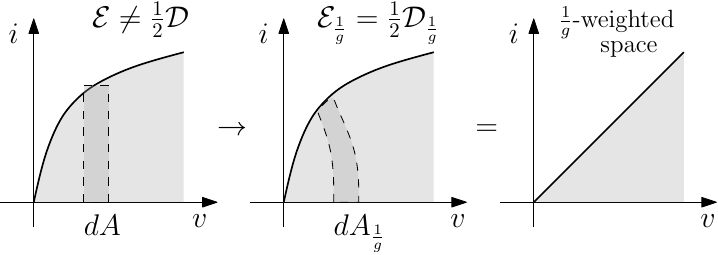}
    \caption{The `linearized' area element due to the metric $\frac{1}{g}$.}
    \label{fig:linearization}
\end{figure}

Crucially, 
$\calD_{\frac{1}{g}} $ and $\calE_{\frac{1}{g}}$ 
continue to determine the behavior of the element.
Using $\grad_{\frac{1}{g}}$ to denote the
Riemannian gradient in $(\real, \frac{1}{g})$,
we have the following result.
\begin{theorem}
\label{theorem:resistive_gradient}
Given the Riemannian manifold $(\real, \frac{1}{g})$,
consider $\calD_{\frac{1}{g}} = v^2 $ and $\calE_{\frac{1}{g}} = \frac{1}{2}v^2$.
Then, the current/voltage relationship of the 
nonlinear resistor \eqref{eq:nonlinear_resistor} at any $v$ satisfies 
\begin{subequations}
\begin{align}
\label{eq:rimeannian_resistor_current1}
\inner{i}{v}_{\frac{1}{g}} = \calD_{\frac{1}{g}}(v)
\end{align}
and
\begin{align}
\label{eq:rimeannian_resistor_current2}
i = \grad_{\frac{1}{g}} \calE_\frac{1}{g}(v).
\end{align}
\end{subequations}
\end{theorem}
\begin{proof}
\begin{subequations}
For \eqref{eq:rimeannian_resistor_current1}, computing
$\inner{i}{v}_{\frac{1}{g}}$ at $v$ for $i = g(v)v$ 
we get 
 \begin{align}
\inner{i}{v}_{\frac{1}{g}} = \frac{1}{g(v)}g(v)vv = v^2.
\end{align}
 For \eqref{eq:rimeannian_resistor_current2}, 
 recall that the Riemannian gradient at $v$
 satisfies 
\begin{align}
\label{eq:nonlinear_resistor_manifold_grad1}
\inner{\grad_{\frac{1}{g}} \calE_{\frac{1}{g}}(v)}{v}_{\frac{1}{g}} \!\! = D_v \calE_{\frac{1}{g}} (v) = v^2.
\end{align}
That is, 
\begin{align}
\label{eq:nonlinear_resistor_manifold_grad2}
\grad_{\frac{1}{g}} \calE_{\frac{1}{g}} (v) = g(v)v,
\end{align}
which establishes \eqref{eq:rimeannian_resistor_current2}.
\end{subequations}
\end{proof}

\subsection{From resistive to memristive elements}
\label{section:memresistive_elements_characterization}

We will now proceed in analogy with the resitive case, replacing the instantaneous inner product 
by the signal inner product 
\begin{equation}
\inner{\bv_1}{\bv_2} = 
\int_{-\infty}^T \bv_1(t) \bv_2(t)  \dd{t}
\end{equation}
where $\bv_1, \bv_2$ are generic signals in $\calL_{2(-\infty,T]}$. 

In analogy with the resistive case, we endow $\calL_{2(-\infty,T]}$ with the structure of a
Riemannian manifold $(\calL_{2(-\infty,T]},\frac{1}{\bg})$. 
The inner product 
$\inner{\cdot}{\cdot}_{\frac{1}{\bg}}$
at $\bv \in \calL_{2(-\infty,T]}$ 
is given by
\begin{equation}
\inner{\bv_1}{\bv_2}_{\frac{1}{\bg}} = 
\int_{-\infty}^T \frac{\bv_1(t) \bv_2(t)}{\bg(\bv)(t)}  \dd{t}
\end{equation}
where $\bv_1,\bv_2 \in T_\bv \calL_{2(-\infty,T]}$
are generic tangent vectors. 
\eqref{eq:assume_invertibility} guarantees the convergence of the integral.

We can then derive the $\frac{1}{\bg}$-weighted dissipation
and $\frac{1}{\bg}$-weighted co-content of the element, as follows.

Given the memristor \eqref{eq:memristor},
under Assumption \ref{assumption:causality_and_invertibility},
for all $\bv \in \calL_{2(\infty,T]}$, 
\begin{subequations}
\label{eq:memresisto_dissipation_and_cocontent}
\begin{align}
\label{eq:memristor_dissipation}
\calD_{\frac{1}{\bg}}(\bv) 
&:= \inner{\bi}{\bv}_{\frac{1}{\bg}} = \inner{\bv}{\bv}. 
\end{align}
Furthermore, for all $\bar{\bv} \in \calL_{2(\infty,T]}$, 
\begin{align}
\label{eq:memristor_cocontent}
\calE_{\frac{1}{\bg}}(\bar{\bv}) 
&:= \!\int_0^{\bar{\bv}}\!\!\!\!\inner{\bi}{\dd{\bv}}_{\frac{1}{\bg}} 
= \frac{1}{2}\inner{\bar{\bv}}{\bar{\bv}}.
\end{align}
\end{subequations}

\begin{subequations}
The identity in \eqref{eq:memristor_dissipation} is straightforward.
\begin{align}
\label{eq:memristor_dissipation_proof}
 \inner{\bi}{\bv}_{\frac{1}{\bg}} 
&= \int_{-\infty}^T \frac{1}{\bg(\bv)(t)} \bg(\bv)(t) \bv(t)^2\dd{t}
\nonumber \\
&= \int_{-\infty}^T \bv(t)^2 \dd{t} 
= \inner{\bv}{\bv}. 
\end{align}
For \eqref{eq:memristor_cocontent}, we have
\begin{align}
\label{eq:memristor_cocontent_proof}
\int_0^{\bar{\bv}}  \inner{\bi}{\dd{\bv}}_{\frac{1}{\bg}} 
&= \int_0^{\bar{\bv}}  \int_{-\infty}^T \frac{1}{\bg(\bv)(t)} \bg(\bv)(t) \bv(t) \dd{\bv}(t) \dd{t}\nonumber \\
&= \int_0^{\bar{\bv}}  \int_{-\infty}^T \bv(t) \dd{\bv}(t) \dd{t} 
= \int_0^{\bar{\bv}}  \inner{\bv}{\dd{\bv}} \nonumber \\
&= \frac{1}{2} \int_0^{\bar{\bv}} \bD_{\dd{\bv}}\inner{\bv}{\bv} 
= \frac{1}{2}\inner{\bar{\bv}}{\bar{\bv}} 
\end{align}
where $\bD_{\dd{\bv}}$ is the directional derivative in the direction $\dd{\bv}$.
\end{subequations}
\begin{remark}
it is instructive to revisit the derivation \eqref{eq:memristor_cocontent_proof}
using a generic but explicit parameterization of the voltage path.
For instance, for $s\in[0,1]$, consider any piecewise differentiable parameterization
$\bv_s \in \calL_{2(-\infty,T]}$ such that
$\bv_0 = 0$ and $\bv_1 = \bar{\bv}$. For each $s$, the memristive
current $\bi_s$ satisfies $i_s = g(\bv_s) v_s$.
Then, 
\begin{align}
\int_0^{\bar{\bv}} \!\!\!\! \inner{\bi}{\dd{\bv}}_{\frac{1}{\bg}} 
& = \! \int_0^1 \! \inner{\bi_s}{\frac{\dd{}}{\dd{s}}\bv_s}_{\frac{1}{\bg}} \!\! \dd{s} 
= \int_0^1 \!\! \inner{\bv_s}{\frac{\dd{}}{\dd{s}}\bv_s} \dd{s} \nonumber \\  
& = \frac{1}{2}\int_0^1 \! \frac{\dd{}}{\dd{s}} \inner{\bv_s}{\bv_s} \dd{s}  
= \frac{\inner{\bv_1}{\!\bv_1}}{2} - \frac{\inner{\bv_0}{\!\bv_0}}{2} \nonumber \\
& = \frac{1}{2} \inner{\bar{\bv}}{\bar{\bv}}. 
\end{align}
\end{remark}

\eqref{eq:memresisto_dissipation_and_cocontent} show an
exact symmetry between resistors and memristors, modulo the change  from $\real$ to $\calL_{2(-\infty,T]}$. 
Dissipation $\calD_{\frac{1}{\bg}}$ and co-content $\calE_{\frac{1}{\bg}}$
are constructed in the same way and are represented by  the same intrinsic expressions.
Crucially, they also determine the memristive behavior. That is,
we can extend Theorem \ref{theorem:resistive_gradient} 
to memristors.

\begin{theorem}
\label{theorem:memresistive_gradient}
Given the Riemannian manifold $(\calL_{2(-\infty,T]}, \frac{1}{\bg})$,
consider $\calD_{\frac{1}{\bg}} = \inner{\bv}{\bv} $ and 
$\calE_{\frac{1}{\bg}} = \frac{1}{2}\inner{\bv}{\bv}$.
Then, the current/voltage relationship of the 
memristor \eqref{eq:memristor} at any $\bv$ satisfies 
\begin{subequations}
\begin{align}
\label{eq:memristor_current1}
\inner{\bi}{\bv}_{\frac{1}{\bg}} = \calD_{\frac{1}{\bg}}(\bv)
\end{align}
and
\begin{align}
\label{eq:memristor_current2}
\bi = \bgrad_{\frac{1}{\bg}} \calE_\frac{1}{\bg}(\bv).
\end{align}
\end{subequations}
\end{theorem}
\begin{proof}
\begin{subequations}
For \eqref{eq:memristor_current1}, 
\begin{align}
 \inner{\bi}{\bv}_{\frac{1}{\bg}} \
&= \!\int_{-\infty}^T \frac{\bi(t)\bv(t)}{\bg(\bv)(t)}  \dd{t}
\stackrel{\mbox{\eqref{eq:memristor}}}{=} \!
\int_{-\infty}^T \!\! \bv(t)^2 \dd{t} 
= \calD_{\frac{1}{\bg}}(\bv). 
\end{align}
 For \eqref{eq:memristor_current2}, 
 recall that 
\begin{align}
\label{eq:memristor_grad_identities}
\inner{\bgrad_{\frac{1}{\bg}} \calE_{\frac{1}{\bg}}(\bv)}{\bv}_{\frac{1}{\bg}} \!\! = \bD_\bv \calE_{\frac{1}{\bg}} (\bv) = \inner{\bv}{\bv},
\end{align}
where the first identity follows from the definition of
Riemannian gradient at $v$, and the second identity follows from
the application of the directional derivative. In fact,
the directional derivative of $\calE_{\frac{1}{\bg}}$ in the 
generic direction $\bv\in \calL_{2(-\infty,T]}$ computed at $\bv\in \calL_{2(-\infty,T]}$
satisfies
\begin{align}
\bD_\bv \calE_{\frac{1}{\bg}}(\bv)
&= \lim_{\epsilon\to 0}  \frac{1}{2\epsilon}\int_{-\infty}^T \!\!\! (\bv(t)+ \epsilon \bv(t))^2 
- v^2 \dd{t} \nonumber \\
&= \lim_{\epsilon\to 0}  \frac{1}{\epsilon}\int_{-\infty}^T \!\!\! \epsilon \bv(t)^2 
+ \epsilon^2\bv(t)^2 \dd{t} \nonumber \\
&= \inner{\bv}{\bv}.
\end{align}
Therefore,
\begin{align}
\bgrad_{\frac{1}{\bg}} \calE_{\frac{1}{\bg}} (\bv) = \bg(\bv)(\cdot)\bv(\cdot),
\end{align}
which establishes \eqref{eq:memristor_current2}.
\end{subequations}
\end{proof}

We conclude the section by observing that the gradient characterization of memristive elements reduces to the gradient characterization of resistive elements in the limit of a memoryless memconductance.

\subsection{Gradient memristive networks}
Definitions and results in the previous section were given for single-input single-output (SISO) elements, but they are easily extended to 
the multiple-input multiple-output (MIMO) case.
We take $\bV = [\bv_1 \dots \bv_n ]^\tran \in \calL_{2(-\infty,T]}^n$ and $\bI = [\bi_1 \dots \bi_n ]^\tran \in \calL_{2(-\infty,T]}^n$, for some $n \in \mathbb{N}$. 
Then, a memristive network is represented by \eqref{eq:memristor}, which reads 
\begin{equation}
\label{eq:memristor_network}
\bI(t) \, = \, \bG(\bv)(t) \, \bV(t),
\end{equation}
where
\begin{align}
\label{eq:memnetwork_coordinates}
\bG \in \calL_{2(-\infty,T]}^n \to \calL_{\infty}^{n\times n}.
\end{align}
That is, at time $t$, the $\bi_k$ current is given by
\begin{align} 
i_k = \sum_{j=1}^n G_{kj}(\bV) v_j, 
\qquad k \in \{1,\dots,n\}.
\end{align}

We retain the two fundamental assumptions of fading memory and positive memconductance, which for a network means
\begin{assume}  
\label{assumption:memnetwork_causality_and_invertibility}
For all $\bv \in \calL_{2(-\infty,T]}^n$ 
and $t\in (-\infty,T]$,
\begin{subequations}
\begin{align}
 \label{eq:memnetwork_assume_invertibility} 
\bG(\bv)(t)^\tran = \bG(\bv)(t) & \succ \varepsilon 
\end{align}
for some $0 < \varepsilon \in \real$.
\end{subequations}
\end{assume}
The memristive network is thus passive. 
Given the standard inner product 
\begin{equation}
\inner{\bv_1}{\bv_2} = 
\int_{-\infty}^T\bv_1(t)^\tran \bv_2(t)  \dd{t}
\end{equation}
where $\bv_1, \bv_2 \in \calL_{2(-\infty,T]}^n$, 
for all $\bv\in\calL_{2(-\infty,T]}^n$, we have
\begin{align}
\inner{\bi}{\bv} = \int_{-\infty}^T \bv(t)^\tran \bG(\bv)(t) \bv(t)  \dd{t} \geq 0.
\end{align}

To extend the results of Section \ref{section:memristive_elements}, 
we consider the 
Riemannian manifold $(\calL_{2(-\infty,T]}^n,\frac{1}{\bG})$
where $\frac{1}{\bG}$ denotes the
Riemannian tensor that satisfies
\begin{equation}
\frac{1}{\bG}(\bv)(t) = \bG(\bv)(t)^{-1},
\end{equation}
for all $\bv\in \calL_{2(-\infty,T]}^n$ and all $t\in \calL_{2(-\infty,T]}$. 
At each $\bv \in \calL_{2(-\infty,T]}^n$,
this leads to the inner product
\begin{equation}
\inner{\bv_1}{\bv_2}_{\frac{1}{\bG}} = 
\int_{-\infty}^T \bv_1(t)^\tran \bG(\bv)(t)^{-1} \bv_2(t) \dd{t}
\end{equation}
where $\bv_1,\bv_2 \in T_\bv \calL_{2(-\infty,T]}^n$
are generic tangent vectors. 

In this setting, 
dissipation $\calD_{\frac{1}{\bG}}$ and
co-content $\calE_{\frac{1}{\bG}}$ are given by 
\eqref{eq:memristor_dissipation} and
\eqref{eq:memristor_cocontent}, respectively
(with notation adapted to the metric $\frac{1}{\bG}$).
In addition, 
Theorem \ref{theorem:memresistive_gradient}
holds for memristive networks,
by replacing
$(\calL_{2(-\infty,T]},\frac{1}{\bg})$
with
$(\calL_{2(-\infty,T]}^n,\frac{1}{\bG})$.

\section{Gradient memRC circuits}

\subsection{Gradient modeling of RC circuits}
Consider a simple RC circuit 
obtained by attaching a capacitor to a resistive element and a current source.
Considering voltages $v\in\real$ and currents $i\in \real$,
and taking the resistive co-content
$\calE = \frac{1}{2} g v^2$, $g\in \real_{>0}$,
the resulting equation is the gradient system
\begin{align}
C \dot v 
=  - gv + i  
=   -\grad  \calE(v) + \grad \inner{v}{i}.
\end{align}
$C$ is the capacitance.
For constant currents $i$, the voltage descends the potential
\begin{equation}
\calP(v) = \calE(v) - \inner{v}{i},
\end{equation} and eventually
converges to a critical point (if the potential is bounded from below).
Namely, for $\gamma = \frac{1}{C}$,
\begin{align}
\dot{\calP}(v) 
= D_{\dot{v}} \calP(v)= 
\inner{\grad \calP(v)}{\dot{v}}
= - \gamma \grad \calP(v)^2.
\end{align}

The gradient nature of RC circuits is a general property and has long been acknowledged in nonlinear circuit modelling. Key references include the early work of Brayton and Moser \cite{brayton1964}, the formulation of input-output gradient systems \cite{brockett1977}, and the important connections between input-output gradient and hamiltonian modeling \cite{vanderschaft2011}.  The expository article \cite{Jeltsema2009} and the textbook \cite{vanderschaft2014} provide an extensive treatment both for RC and RLC circuits. 

Networks containing only capacitors, resistive elements, and current sources, 
obey the same gradient dynamics.
Take $V,I \in \realn$ for some $n \in \mathbb{N}$
and consider the resistive co-content $\calE = \frac{1}{2} V^\tran G V$
with $G \succ 0$ diagonal. We get
\begin{align}
    C \dot V =  - \grad  \calE(V) + \grad \inner{V}{I}.
\end{align}
where $C \succ 0$ is a diagonal matrix of capacitances.
The $j$-th element $v_j$ of $V$ satisfies
\begin{align}
    C_j \dot v_j =  - \grad  \calE_{j} (v_j) + \grad \inner{v_k}{i_k}.
\end{align}
for $\calE_{j} (v_j) = \frac{1}{2} g_{j} v_j^2$.
In fact, the co-content is summed over all the resistive elements: 
$\calE (V) = \sum_{j=1}^n \calE_{j} (v_j)$.

As explained in \cite{Jeltsema2009}, the equation can be regarded as a (very) special case of Lagrangian or Hamiltonian modeling where both the co-Lagrangian ${\calL}^*(V)$ and the co-Hamiltonian ${\calH}^*(V)$ reduce to the total stored
capacitive co-energy
\begin{align}
 {\calL}^*(V) = {\calH}^*(V) = \sum_{j=1}^n \frac{1}{2} C_jv_j^2. 
\end{align}
The circuit equations reduce to
\begin{align}
\frac{\dd{}}{\dd{t}}  \grad {\calH}^*(V) 
=  - \grad  \calE(V) + \grad \inner{V}{I}.
\end{align}

A notable example of nonlinear RC circuits in neural networks modeling is provided by Hopfield neural networks \cite{hopfield1984}. 
The circuit is made of $n$ neurons that obey the differential equation
\begin{align}
C_k \dot v_k = -\sum_{j=1}^n w_{kj} \Phi_j (v_j) + i_k, \; 1,\dots,n 
\end{align}
which can be written in the vector form
\begin{align}
\label{eq:hopefield}
C \dot V = -W \Phi(V) + I 
\end{align}
where $V, I \in \realn$, $C = \mathrm{diag}(C_1,\dots,C_n)$ 
$W$ is a symmetric weight matrix with elements $w_{kj}$, and 
$\Phi(V)$ is a vector with elements $\Phi_k(v_k)$. 
It is well-known that Hopfield networks are gradient systems. 

If we assume $W \succ 0$, we can interpret 
the Hopfield potential as the resistive co-content
\begin{align}
\calE(V) = \frac{1}{2}\Phi(V)^TW\Phi(V).
\end{align}
In addition, if the element $\Phi_k(v_k)$ are monotone, that is, 
$\frac{\dd{}}{\dd{v_k}} \Phi_k(v_k) > 0$ for all $k$, 
we can define the metric
\begin{align}
G(V) = \mathrm{diag}\left(\frac{\dd{}}{\dd{v}_1} \Phi_1(v_1), \dots, \frac{\dd{}}{\dd{v}_n} \Phi_n(v_n)\right).
\end{align}
Then, \eqref{eq:hopefield} reads
\begin{align}
\label{eq:hopefield2}
 C \dot{V} = -\grad_{G} \calE(V) + \grad \inner{V}{I}.
\end{align}

\subsection{From RC to memRC circuits}

To construct the gradient model of a  memRC circuit, we will now mimick the construction of Section \ref{section:memresistive_elements_characterization}, replacing the resistive gradient vector field in the space of voltages by memristive gradient operators in the space of past voltage trajectories $\calL_{2(-\infty,T]}$.

An elementary memRC circuit is obtained by attaching a capacitor to a memristive 
element and a current source.
The memristive element satisfies 
\eqref{eq:memristor}, for some given
memconductance $\bg$. Then, 
given the cocontent
$\calE_{\frac{1}{\bg}}(\bv) = \frac{1}{2} \inner{\bv}{\bv}$,
the resulting equation is 
\begin{equation}
\label{eq:memRC1}
    C \dot{v} = - \grad_{\frac{1}{\bg}} \calE_{\frac{1}{\bg}}(\bv) + \grad\inner{\bi}{\bv},
\end{equation}
where we have dropped the dependence on time for simplicity, using non-bold symbols
to denote signals and operators evaluated at time $t$ (e.g., $v$ instead $\bv(t)$,
and $\grad$ instead of $\bgrad\dots(t)$).

For any given current 
$\bi \in \calL_{2(-\infty,T]}$, \eqref{eq:memRC1} can be considered as an
equation with unknown $\bv \in \calL_{2(-\infty,T]}$. Solutions 
are the zero of 
\begin{equation}
\label{eq:memRC2}
    \left( C \frac {\dd{}}{\dd{t}} + \bgrad_{\frac{1}{\bg}} \calE_{\frac{1}{\bg}} (\cdot) - \bgrad\inner{\cdot}{\bi} \right)(\bv)  = 0.
\end{equation}
\eqref{eq:memRC2} makes clear that every
solution $\bv$ of the memRC circuit must balance three main factors, namely the
stored capacitive energy, the ${\frac{1}{\bg}}$-weighted resistive energy, and the supplied energy.
The energy
stored in the capacitor at time $t$ is 
$\calS_C(\bv(t)) = \frac{1}{2} C \bv(t)^2$. We have
\begin{align}
S_C(\bv(T)) - S_C(\bv(-\infty))
&= \inner{\bv}{C\dot{\bv}} \nonumber \\
&= \inner{\bv}{\bgrad_\bg \calE(\bv) - \bgrad\inner{\bv}{\bi}} \nonumber \\
&= \inner{\bv}{\bgrad_\bg \calE(\bv) - \bi}
\end{align}
that is, every solution $\bv$ satisfies
\begin{align}
\label{eq:energy_balance}
\inner{\bv}{C\dot{\bv} + \bgrad_\bg \calE(\bv) - \bi} = 0,
\end{align}
and the latter further simplifies to
\begin{align}
\label{eq:simplified_energy_balance}
\inner{\bv}{\bgrad_\bg \calE(\bv) - \bi} = 0.
\end{align}
for all closed voltage trajectories, 
namely these trajectories that satisfy $\lim_{t\to -\infty} \bv(t) = \lim_{t\to T} \bv(t)$ (for example, a trajectory that start and ends at `rest').

\eqref{eq:memRC2}
shows that the memRC circuit
models a {\it gradient behavior} in the space of past trajectories.
The solutions of the memRC circuit are critical points of quadratics, as in the classical resistive case.

\section{Hodgkin-Huxley model}

We use memristive modeling
to represent the neuronal dynamics of
conductance-based models, such as
Hodgkin-Huxley neurons \cite{Keener2009a}. 
As illustrated in Figure \ref{fig:HH_neuron},
the equivalent circuit of a neuron 
is given by a capacitor, whose
voltage is regulated by 
parallel branches of a memconductances
in series with batteries. The
voltage equation reads
\begin{subequations}
\label{eq:cb_neuron}
\begin{align}
\label{eq:cb_neuron_membrane}
C \dot{v} = -i_L - i_E - i_I + i
\end{align}
where $C$ is the membrane capacitance.
The currents are organized into four functional groups, namely \emph{leak} $L$, \emph{excitatory} $E$, \emph{inhibitory} $I$, and \emph{externally supplied}, 
We model leak, excitatory, and inhibitory currents
with memristive elements \eqref{eq:memristor},
\label{eq:cb_neuron_current}
\begin{align}
i_L &= g_L(\bv-\bar{\bv}_L) (v - \bar{v}_L) 
\label{eq:cb_neuron_leak}\\
i_E &= g_E(\bv-\bar{\bv}_E) (v - \bar{v}_E)  
\label{eq:cb_neuron_excitatory}\\
i_I &= g_I(\bv-\bar{\bv}_I) (v - \bar{v}_I)  
\label{eq:cb_neuron_inhibitory}
\end{align}
\end{subequations}
where
$\bg_L,\bg_E,\bg_I \in \calL_{2(-\infty,T]} \to \calL_\infty$ satisfy Assumption \ref{assumption:causality_and_invertibility},
and $\bar{v}_L$, $\bar{v}_E$, $\bar{v}_I$ represent the constant voltage of the batteries in series to the memconductances.

\begin{figure}[htbp]
    \centering
    \includegraphics[width=0.66\columnwidth]{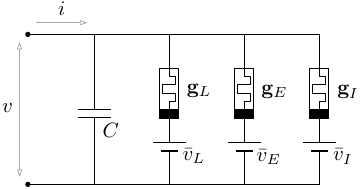}
    \caption{Conductance-based neuron equivalent circuit.}
    \label{fig:HH_neuron}
\end{figure}

\begin{exx}
In the Hodgkin-Huxley model of a neuron we would have $\bar{v}_I < \bar{v}_L < \bar{v}_E$. Furthermore,
\begin{subequations}
\label{eq:HH_conductances}
\begin{align}
g_L = \bar{g}_L
\quad
g_E = \bar{g}_E m^3h
\quad
g_I = \bar{g}_i n
\end{align}
for constant $\bar{g}_L,\bar{g}_E,\bar{g}_I \in \real_{>0}$ and 
\begin{align}
\tau_m(v)\dot{m} &= -m + m_\infty(v)   \\
\tau_h(v)\dot{h} &= -h + h_\infty(v) \\
\tau_n(v)\dot{n} &= -n + n_\infty(v),
\end{align}
\end{subequations}
where 
$\tau_{m}, \tau_{h}, \tau_{n}$ are
voltage-dependent time constants,
and
$m_\infty, h_\infty, n_\infty$ 
are static voltage-dependent nonlinear functions, \cite{Keener2009a}.  \hfill $\lrcorner$
\end{exx}

Considering the resistive co-contents
\begin{align}
\label{eq:cb_neuron_three_cocontents}
\calE_{{\frac{1}{\bg_L}}} =
\calE_{{\frac{1}{\bg_E}}} = 
\calE_{{\frac{1}{\bg_I}}} = 
\frac{1}{2}\inner{\bv}{\bv},
\end{align}
we can rewrite \eqref{eq:cb_neuron}
as 
\begin{align}
C \dot{v} 
= & 
-\grad_{\frac{1}{\bg_L}} \calE_{{\frac{1}{\bg_L}}}(\bv-\bar{\bv}_L)
-\grad_{\frac{1}{\bg_E}} \calE_{{\frac{1}{\bg_E}}}(\bv-\bar{\bv}_E) \nonumber \\
& -\grad_{\frac{1}{\bg_I}} \calE_{{\frac{1}{\bg_I}}}(\bv-\bar{\bv}_I)
+ \grad\inner{\bv}{\bi}.
\label{eq:cb_neuron2}
\end{align}
Take now
$\one^\tran = [1, 1, 1]$
and define
\begin{subequations}
\begin{align}
\bar{\bV}^\tran &= 
\mymatrix{ccc}
{\bar{\bv}_L & \bar{\bv}_E & \bar{\bv}_I}, \\
\frac{1}{\bG} &=
\mathrm{diag}\left(
\frac{1}{\bg_L},
\frac{1}{\bg_E},
\frac{1}{\bg_I}
\right), \\
\calE_{\frac{1}{\bG}} & = 
\mathrm{diag}\left(
\calE_{{\frac{1}{\bg_L}}},
\calE_{{\frac{1}{\bg_E}}},
\calE_{{\frac{1}{\bg_I}}}
\right).
\end{align}
Then, \eqref{eq:cb_neuron2} reads
\begin{align}
\label{eq:gradient_cb_neuron}
C \dot{v} 
= \one^\tran \grad_{\frac{1}{\bG}}
\calE_{\frac{1}{\bG}}(\one \bv - \bar{\bV}) 
+ \grad\inner{\bv}{\bi}.
\end{align}
\end{subequations}
Solutions to \eqref{eq:gradient_cb_neuron},
such as a voltage spike,
must balance stored capacitive energy, the $\frac{1}{\bG}$-weighted resistive energy, and the supplied energy \eqref{eq:energy_balance}. In fact, 
any spike starts and ends at rest, therefore satisfies the simplified balance \eqref{eq:simplified_energy_balance}.

The gradient form \eqref{eq:gradient_cb_neuron}
shows that the complex nonlinear 
voltage response of a neuron to an external current is a zero of the voltage equation \eqref{eq:memRC2}.
We can thus use numerical algorithms to solve
this zero-finding problem. This is beyond the scope of the current paper but a brief illustration for the Hodgkin-Huxley neuron \eqref{eq:cb_neuron} is provided in the following example.

\begin{exx}
We solve \eqref{eq:cb_neuron} with a two-step
iteration. 

\noindent\textbf{Compute(g):} for a given voltage $\bv$, we compute $\bg_L(\bv-\bar{\bv}_L)$, $\bg_E(\bv-\bar{\bv}_E)$, and 
$\bg_I(\bv-\bar{\bv}_I)$ using 
\eqref{eq:HH_conductances}.

\noindent \textbf{Compute(v):} we compute the new voltage $\bv^+$ 
by solving \eqref{eq:cb_neuron} using the 
values for the memconductances computed above and the input current $\bi$.

The iteration is initialized with constant signals $\bv$, $\mathbf{m}$, $\mathbf{h}$, $\mathbf{n}$ related to
the circuit at `rest'. These correspond to the circuit equilibrium for $\bi = 0$. We consider 
an interval of time $0 \mbox{ ms} \leq t \leq 50 \mbox{ ms}$. Each signal is represented by $500$ samples.

The input current is $\bi(t) = 10 mA$ for $20 \mathrm{ms} \leq t \leq 22 \mathrm{ms}$ and $0$ otherwise,  as shown in Figure \ref{fig:HH_sim}, bottom-right plot.
The top plots of Figure \ref{fig:HH_sim} show the  trajectories of the circuit obtained by standard numerical integration. The bottom-left plot shows the voltage trajectory of the circuit obtained via iteration. 
The iteration converges to the circuit solution in five steps. \hfill $\lrcorner$
\end{exx}

\begin{figure}[htbp]
    \centering
    \vspace{1mm}
    \includegraphics[width=0.49\columnwidth]{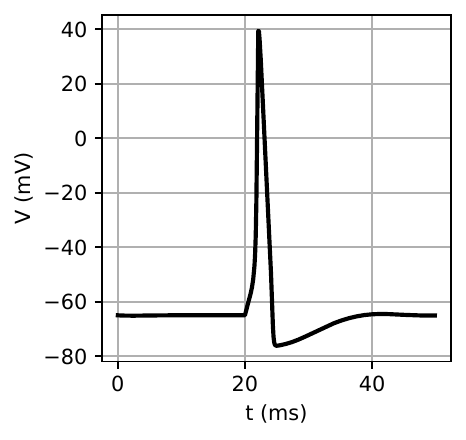}
    \includegraphics[width=0.49\columnwidth]{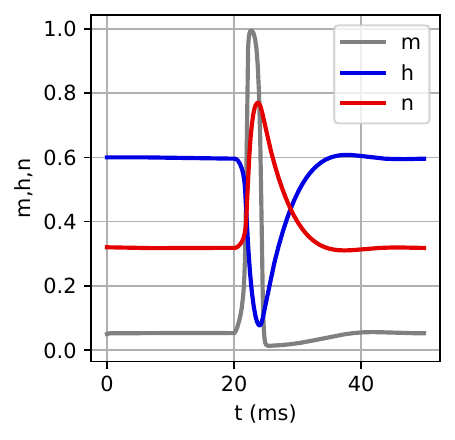}
    \includegraphics[width=0.49\columnwidth]{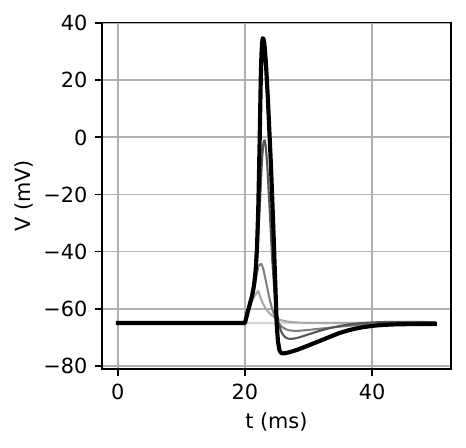}
    \includegraphics[width=0.49\columnwidth]{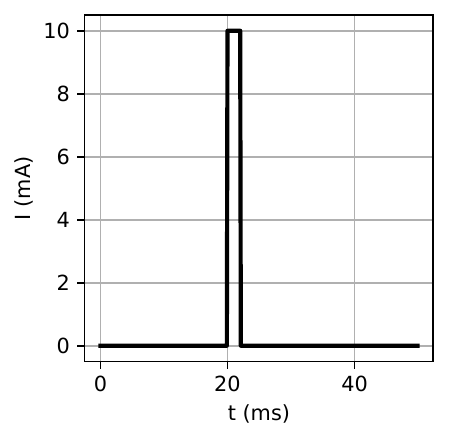}
    \caption{Hodgkin-Huxley spike. \textbf{Top:} ODE solution via nmumerical integration. \textbf{Bottom|left:} ODE solution via alternate iteration. \textbf{Bottom|right:} external current.}
    \label{fig:HH_sim}
\end{figure}

\section{Discussion}
This paper has revisited the classical memristive model of Chua and Wang \cite{chua1976} from a gradient modeling perspective. The key observation is that memristive elements are very similar to resistive elements, provided we consider the current-voltage relationship in the signal space of past trajectories rather than as an pointwise relationship in time.

The consequence is that we can also regard neurophysiological and neuromorphic behaviors as very similar to RC circuits, provided we analyze the circuit behavior as a relationship between current and voltage trajectories. For a given current trajectory, the voltage trajectory is the critical point of an energy potential. 

To the best of the authors knowledge, this viewpoint is novel and it paves the way to new methodologies for the analysis and design of neurophysiological and neuromorphic behaviors. 

The proposed viewpoint is aligned with the operator-theoretic perspective recently developed to simulate neuromorphic circuits by regarding their solutions as zeros of monotone (or difference of monotone) operators and by applying splitting methods that exploit the parallel topology of the circuit \cite{chaffey2023,shahhosseini2024}.

More generally, an important outcome of the proposed gradient modelling approach is a decoupling between the memconductance modeling and the memRC circuit modeling. The memconductance is a fading memory operator from voltage to memconductances, which can be efficiently modelled as a ``feedforward'' nonlinear convolution operator. The memRC circuit is a ``feedback'' nonlinear circuit, but simple to analyze because of its gradient nature. The decoupling matches the separation of two physical domains: the biochemical process that governs the gating of the conductance, and the electrical process that interconnects the neurons via current flows. The first process is complex but feedforward. The second process is recurrent but simple.

Future research will explore how to scale this modeling framework to large neural networks. This will be achieved by considering finite dimensional subspaces of $\calL_{2(-\infty,T]}^n$ in order to analyse the circuit {\it at scale}.

\balance
\bibliographystyle{IEEEtran}

\end{document}